\documentclass[letterpaper, 10 pt, conference]{ieeeconf}  

\IEEEoverridecommandlockouts                              
\usepackage{cite}
\usepackage{amsmath,amssymb,amsfonts}
\usepackage{algorithmic}
\usepackage{graphicx}
\usepackage{textcomp}
\usepackage{url}
\usepackage{color, xcolor} 
\usepackage{epsfig} 
\usepackage{amsmath, bbm, bm} 
\usepackage{amsfonts, amssymb}
\usepackage{mathrsfs}
\usepackage{anyfontsize}
\usepackage{accents}
\usepackage{tabularx,ragged2e}
\usepackage{multirow}
\usepackage{enumerate}
\usepackage{makecell}
\usepackage{float}
\usepackage{wrapfig}
\usepackage{caption}
\usepackage{subcaption}
\usepackage{makecell}
\usepackage{tikz}
\usepackage{tikz-3dplot}
\usepackage{pgfplots}
\pgfplotsset{compat=newest}
\usepackage{mathtools}

\usetikzlibrary{decorations.markings}
\usetikzlibrary{patterns}

\newtheorem{proposition}{Proposition}

\newtheorem{corollary}{Corollary}
\newtheorem{theorem}{Theorem}

\newtheorem{definition}{Definition}

\DeclareMathOperator*{\argmax}{arg\,max}

\newcommand\TinyMatrix[1]{{%
		\small\arraycolsep=0.3\arraycolsep\ensuremath{\begin{bmatrix}#1\end{bmatrix}}}}
\newcommand{\ubar}[1]{\underaccent{\bar}{#1}}

\title{\LARGE \bf Population Games With Erlang Clocks: \\ Convergence to Nash Equilibria For Pairwise Comparison Dynamics
 }

\author{Semih Kara \and Nuno C. Martins \and Murat Arcak
\thanks{Kara and Martins' work was supported by AFOSR Grant FA9550-19-1-0315 and NSF Grant CNS-2135561. Arcak's work was supported by NSF Grant CNS-2135791.}
\thanks{Semih Kara and Nuno C. Martins are with the ECE Department, and the ISR, University of Maryland, College Park, MD 20742, USA.
        {\tt\small \{skara, nmartins\}@umd.edu.}}%
\thanks{Murat Arcak is with the EECS Department, University of California, Berkeley, CA 94720, USA.
        {\tt\small arcak@berkeley.edu.}}%
}

\begin{document}

\maketitle

\begin{abstract}
    The prevailing methodology for analyzing population games and evolutionary dynamics in the large population limit assumes that a Poisson process (or clock) inherent to each agent determines when the agent can revise its strategy. Hence, such an approach presupposes exponentially distributed inter-revision intervals, and is inadequate for cases where each strategy entails a sequence of sub-tasks (sub-strategies) that must be completed before a new revision time occurs. This article proposes a methodology for such cases under the premise that a sub-strategy's duration is exponentially-distributed, leading to Erlang distributed inter-revision intervals. We assume that a so-called pairwise-comparison protocol captures the agents' revision preferences to render our analysis concrete. The presence of sub-strategies brings on additional dynamics that is incompatible with existing models and results. Our main contributions are twofold, both derived for a deterministic approximation valid for large populations. We prove convergence of the population's state to the Nash equilibrium set when a potential game generates a payoff for the strategies. We use system-theoretic passivity to determine conditions under which this convergence is guaranteed for contractive games.
\end{abstract}

\section{Introduction}

Population games and evolutionary dynamics have been used as a tractable framework to model the strategic interactions in populations with large numbers of agents~\cite{Sandholm2010Population-Game,Park2018Payoff-Dynamic-,Park2019From-Population}. In this framework, each agent follows one strategy at a time, chosen from a finite set available to the population. At any time, each available strategy has a payoff specified by a population game. The agents repeatedly revise their strategies at the so-called revision times governed by a stochastic process (or clock) inherent to each agent. At a revision time, the agent may alter its strategy in response to the payoffs and the strategy profile of the population. When revising their strategies, the agents act according to a probabilistic heuristic specified by a so-called revision protocol, which reflects the population's decision behavior and often has a simple structure. The agents are nondescript, consequently a vector, called population state, whose entries are the proportions of the population following the available strategies suffices to represent the population's strategic profile. 

\subsection{Erlang Revision Times}
\label{subsec:Intro_Erl_Rev_Times}

Existing work assumes that the agents' inter-revision times are exponentially distributed \cite{Sandholm2010Population-Game}, including the generalization in \cite{Kara2021Pairwise-Compar-a}, which allows for strategy-dependent revision rates.

Applications where each strategy entails a sequence of tasks (sub-strategies) an agent must complete before a new revision opportunity occurs are not compatible with the existing theory because the sum of the tasks' service times is generally not exponentially distributed. In this article, our main goal is to propose a methodology to model and analyze the equilibrium stability of population games and evolutionary dynamics for such non-preemptive multi-task applications. Inspired by the queueing literature, we assume that the sub-strategies' service times are independent and exponentially distributed, resulting in Erlang distributed inter-revision times \cite[Chapter~4.2]{Queueing_Systems_Volume_I_Theory_Kleinrock}. 

The generalization that we propose is further motivated by its pertinence to traffic analysis. Congestion games, originally proposed in \cite{Studies_in_the_economics_of_transportation_Beckmann_et_al}, are well-studied examples of population games that represent the strategic environment in road networks. This, combined with the success of revision protocols in capturing myopic decision behavior, make the population games and evolutionary dynamics framework an adequate candidate for analyzing traffic in road networks. Several studies \cite{Regression-based_models_for_bus_dwell_time_Xinkai_Xiaoguang,Exploring_Travel_Time_Distribution_and_Variability_Patterns_Using_Probe_Vehicle_Data_Case_Study_in_Beijing_Chen_et_al,Review_on_Statistical_Modeling_of_Travel_Time_Variability_for_Road-Based_Public_Transport_Buchel_Francesco} find that, in many cases, Gamma and Erlang distributions are suitable representations of the time spent in traffic. Thus, noting that the Erlang and Gamma distributions coincide for certain parameter values, extending the framework to account for revision interarrival times with Erlang distributions enhances its applicability to analyzing traffic in road networks.

\subsection{Deterministic Approximation And Stability Analysis}

Results in the population games and evolutionary dynamics paradigm often seek to ascertain whether the population state converges to a neighborhood of the Nash equilibria of the population game with high probability. Similar to an extensive body of literature \cite{Weibull1995Evolutionary-ga,Sandholm2010Population-Game,Fox2013Population-Game,Park2018Payoff-Dynamic-,Arcak2020Dissipativity-T}, in this paper, we examine the aforementioned problem by employing stochastic approximation theory \cite{Kurtz1970Solutions-of-Or}, \cite[Appendix~12.B]{Sandholm2010Population-Game} and analyzing a system of deterministic differential equations. Following the nomenclature in \cite{Park2018Payoff-Dynamic-}, we refer to this system as the Erlang Evolutionary Dynamics Model.

High probability convergence of the population state to a close vicinity of the Nash equilibria of the population game can't be guaranteed under exponentially distributed inter-revision times, unless further structure is imposed on the game and the revision protocol. Pairwise comparison protocols, in combination with potential and strictly contractive\footnote{Contractive games are also known as stable games. The recent article~\cite{Arcak2020Dissipativity-T} defines weighted-contractive games, which subsume contractive ones. For simplicity, we limit our analysis to the non-weighted case.} games, are known to provide such structure \cite{Sandholm2010Pairwise-compar,Sandholm2001Potential-games,Hofbauer2009Stable-games-an}. Moreover, pairwise comparison protocols have undemanding informational requirements and allow for agents to operate in a decentralized manner \cite{Kara2021Pairwise-Compar-a,Sandholm2010Pairwise-compar}. Hence, we confine our focus to pairwise comparison protocols along with potential and strictly contractive games, and investigate whether the convergence properties of the population state persist when the inter-revision times are extended from exponential to Erlang distributions. Notably, congestion games, which is a part of the motivation in \S\ref{subsec:Intro_Erl_Rev_Times}, are potential games.



\section{Overview of the Framework}
\label{sec:Fram_Desc}

We start by presenting an overview of the population games and evolutionary dynamics paradigm.

\subsection{Agents, Strategies and the Population State}
\label{subsec:Pop_States_and_the_Soc_State}

Consider a set of $N$ agents, where $N$ is large. We refer to the collection of agents as the population.\footnote{Although our results hold in the case of multiple populations (see \cite{Sandholm2010Population-Game} for the setting with multiple populations), for ease of exposition, we assume that there is a single population.} At any time, each agent follows a single strategy from the same set of strategies $\{1,\dots,n\}$. As will be clarified throughout \S\ref{sec:Fram_Desc}, the agents are ``nondescript''; therefore, proportions of agents playing each strategy suffices to characterize the strategy profile of the population. We denote the proportion of agents playing strategy $i\in\{1,\dots,n\}$ at time $t\geq 0$ by $\bar{X}_i(t)$. Moreover, we define $\bar{X}:=\begin{bmatrix} \bar{X}_1 & \dots & \bar{X}_n \end{bmatrix}^T$, and refer to $\bar{X}$ as the population state.

\subsection{Payoffs and Population Games}
\label{subsec:Pop_Games}

At any time $t\geq 0$, each strategy $i\in\{1,\dots,n\}$ is endowed with a payoff $P_i(t)$. We assume that the mechanism assigning these payoffs is a continuously differentiable function $\mathcal{F}:\mathbb{R}^n\to\mathbb{R}^n$, called a population game, that operates on the population state $\bar{X}$. The payoff of strategy $i$ is denoted $\mathcal{F}_i$ and we write $\mathcal{F}:=\begin{bmatrix}\mathcal{F}_1 & \dots & \mathcal{F}_n \end{bmatrix}^T$. Consequently, the payoff vector at time $t$ is given by $P(t)=\mathcal{F}(\bar{X}(t))$. 

The notion of Nash equilibria is defined for population games as given below:
\begin{equation*}
    \mathbb{NE}(\mathcal{F}) := \left \{\xi\in\Delta \ \bigg | \ \xi_i>0 \implies i\in\argmax_{j\in\{1,\dots,n\}}\mathcal{F}_j(\xi) \right \},
\end{equation*}
where $\Delta:=\{\xi\in\mathbb{R}_{\geq 0}^n \ | \ \sum_{i=1}^n\xi_n=1\}$. According to~\cite[Theorem~2.1.1]{Sandholm2010Population-Game}, the set $\mathbb{NE}(\mathcal{F})$ is nonempty.

\subsection{The Revision Paradigm}
\label{subsec:Rev_Paradigm}

Agents repeatedly revise their strategies conforming to a procedure characterized by two components. The first component is the process that specifies the agents' revision times. The second component is the so-called revision protocol, which describes how a revising agent decides on its subsequent strategy.

\subsubsection{Revision Times in the Traditional Framework}
\label{subsubsec:Trad_Rev_Times}

In the traditional framework, revision times of the agents are constructed by assigning the agents independent and identically distributed (i.i.d.) Poisson processes and defining the revision times of an agent as the jump times of its process. Hence, the interarrival times of revisions of agents are i.i.d. exponential random variables \cite{Sandholm2010Population-Game}.

The recent work in \cite{Kara2021Pairwise-Compar-a} extends this construction by allowing the revision times of agents to depend on their current strategies. In this paper we abide by the strategy-independent structure of the traditional framework. However, in \S\ref{sec:Evol_Dyn}, we propose an alternative generalization of the traditional process according to which revision times are determined.

\subsubsection{Revision Protocols}
\label{subsubsec:Rev_Prots}

The revision protocol of the population is a Lipschitz continuous function $\mathcal{T}:\mathbb{R}^n\times\mathbb{R}^n\to \mathbb{R}_{\geq 0}^{n\times n}$ that satisfies $\sum_{j=1}^n \mathcal{T}_{i,j}(\xi,\pi) = \lambda$ for all $i\in\{1,\dots,n\}$, $\pi\in\mathbb{R}^n$ and $\xi\in\Delta$. Intuitively, for any $i,j\in\{1,\dots,n\}$, $\mathcal{T}_{i,j}$ gives the rate with which agents playing strategy $i$ switch to strategy $j$. An important quantity that appears in the stability analysis in \S\ref{sec:Res_Und_Str_Con_Games} is
\begin{align}
\label{eq:c}
    c:=\max_{i\in\{1,\dots,n\},~\xi\in\Delta} \sum_{j=1,~j\neq i}^n \mathcal{T}_{i,j}(\xi,\mathcal{F}(\xi)),
\end{align}
which is a measure of the maximum rate of strategy switching that omits ``self-switches''.

More precisely, the revision protocol determines the subsequent strategy of a revising agent according to the following description. Assume that an agent receives a revision opportunity at time $\bar{t}$. It follows from the definition of revision times that, with probability 1, there is a $t^*$ (strictly) between $\bar{t}$ and the previous revision time of the population. Denoting the strategy at $t^*$ of the revising agent as $i$, the probability of its subsequent strategy being $j$ is assumed to be $\mathcal{T}_{i,j}\big(\bar{X}(t^*),P(t^*)\big)/\lambda$ for any $j\in\{1,\dots,n\}$. Then, the realization of this strategy is assigned to be the revising agent's strategy at time $\bar{t}$.

\subsection{The Evolutionary Dynamics Model}
\label{subsec:EDM}

If the revision times of the agents are as constructed in \S\ref{subsubsec:Trad_Rev_Times}, payoffs are characterized by a population game, and agents decide on their subsequent strategies according to the procedure in \S\ref{subsubsec:Rev_Prots}, then the resulting population state $\bar{X}$ is a pure jump Markov process \cite[Chapter~11]{Sandholm2010Population-Game}. 

Subsequently, an important question is whether $\bar{X}$ converges to $\mathbb{NE}(\mathcal{F})$. As explained in  \cite[Section~5]{Park2018Payoff-Dynamic-} and \cite[Appendix~12.B]{Sandholm2010Population-Game}, stochastic approximation theory provides a methodology to answer this question.
Namely, provided that $N$ is large, from \cite{Kurtz1970Solutions-of-Or} and \cite[Appendix~12.B]{Sandholm2010Population-Game}, it follows that the convergence with high probability of $\bar{X}$ to a neighborhood of $\mathbb{NE}(\mathcal{F})$ can be concluded by verifying that $\mathbb{NE}(\mathcal{F})$ is globally attractive under a deterministic dynamical system. This dynamical system is referred to as the Evolutionary Dynamics Model (EDM) \cite{Park2018Payoff-Dynamic-}, where in this paper we call it the standard EDM to distinguish it from its Erlang counterpart introduced in \S\ref{sec:Evol_Dyn}.

\section{Erlang Evolutionary Dynamics}
\label{sec:Evol_Dyn}

In this section, we propose a generalization of the population games and evolutionary dynamics framework, outlined in \S\ref{sec:Fram_Desc}, by allowing inter-revision times of agents to be independent and identical Erlang random variables.

\subsection{Erlang Revision Times}
\label{subsec:Erl_Rev_Times}

To introduce Erlang distributed inter-revision times to the population games and evolutionary dynamics framework, we follow a construction similar to that in \S\ref{subsubsec:Trad_Rev_Times} and assign the agents i.i.d. arrival processes. We assume the interarrival times of an agent's process to have independent and identical Erlang distributions with rate $\lambda>0$ and parameter $m\in\mathbb{N}$. Consequently, we define the revision times of an agent to be the arrival times of its process.

If the revision times of the agents are as constructed in \S\ref{subsec:Erl_Rev_Times} (instead of the traditional construction in \S\ref{subsubsec:Trad_Rev_Times}) then the resulting population state $\bar{X}$ is a pure jump stochastic process, but not necessarily a Markov process. 
This is undesirable because, to the best of our knowledge, stochastic approximation results similar to that in \cite{Kurtz1970Solutions-of-Or} or \cite[Appendix~12.B]{Sandholm2010Population-Game} do not exist for pure jump processes with Erlang distributed waiting times. So, it is not directly evident how the deterministic approach summarized in \S\ref{subsec:EDM} can be altered to fit the framework with Erlang distributed inter-revision times. Therefore, in the following part, we present an alternative way in which the results in \cite{Kurtz1970Solutions-of-Or} and \cite[Appendix~12.B]{Sandholm2010Population-Game} can be leveraged.

\subsection{Erlang Evolutionary Dynamics}

In what follows, we first characterize the population state in terms of a pure jump Markov process that conforms to the assumptions in \cite{Kurtz1970Solutions-of-Or} and \cite[Appendix~12.B]{Sandholm2010Population-Game}. Then, we apply the aforementioned stochastic approximation results to this process and derive a deterministic dynamical system, which we utilize in the upcoming sections to ascertain the convergence properties of $\bar{X}$.

\subsubsection{Characterizing the Population State in Terms of a Markov Process}

Given a strategy $i\in\{1,\dots,n\}$, let us define $(i,j)$ to be the $j$-th sub-strategy of $i$ for $j\in\{1,\dots,m\}$. Now, consider that an agent who chooses strategy $i$ starts playing $(i,1)$. Suppose that, for any $j\in\{1,\dots,m-1\}$, after playing $(i,j)$ for a period of time the agent transitions to playing $(i,j+1)$. Let the time that the agent spends playing sub-strategy $(i,j)$, for any $j\in\{1,\dots,m\}$, be distributed exponentially with parameter $\lambda$. Furthermore, assume that the agent is given a revision opportunity after it is finished playing sub-strategy $(i,m)$ and that, when it is given an opportunity, the agent chooses its subsequent strategy according to the procedure in \S\ref{subsubsec:Rev_Prots}. Finally, assume that the times spent by agents playing the sub-strategies are i.i.d.

In the scenario described above, interarrival times of agents' revision opportunities are i.i.d. Erlang random variables with parameters $\lambda$ and $m$. As a result, the revision interarrival times constructed above and the ones constructed in \S\ref{subsec:Erl_Rev_Times} have the same joint distributions.

Let us denote the proportion of agents playing sub-strategy $(i,j)$ by $X_{i,j}$ and define
\begin{align*}
&X := \begin{bmatrix} X_{1,1} & \dots & X_{1,m} & \dots & X_{n,1} & \dots & X_{n,m} \end{bmatrix}^T.
\end{align*}
Then, $X$ is a pure jump Markov process, to which the results in \cite{Kurtz1970Solutions-of-Or} and \cite[Appendix~12.B]{Sandholm2010Population-Game} can be applied. Moreover, given any $i\in\{1,\dots,n\}$, $\sum_{j=1}^{m}X_{i,j}$ and $\bar{X}_i$ have the same distribution. Therefore, we can infer the long term behavior of $\bar{X}$ by analyzing $X$.

\subsubsection{The Deterministic Approximation}
\label{subsubsec:Erl_Evol_Dyn}

Provided that the number of agents is large, stochastic approximation theory presents a methodology to analyze the progression of $X$.

Namely, from the results in \cite{Kurtz1970Solutions-of-Or}, \cite[Section~V]{Park2018Payoff-Dynamic-} and \cite[Appendix~12.B]{Sandholm2010Population-Game} we have for any $T>0$ and $\epsilon>0$ that
\begin{align*}
    \lim_{N\to\infty}\mathbb{P}\Bigg(\sup_{t\in[0,T]}\|X(t)-x(t)\|<\epsilon \Bigg) = 1,
\end{align*}
where $$x:=\begin{bmatrix}x_{1,1} & \dots x_{1,m} & \dots & x_{n,1} & \dots & x_{n,m} \end{bmatrix}^T$$ is the solution with $x(0)=X(0)$ of the system of differential equations given for $1 \leq i \leq n$ and $2 \leq \ell \leq m$ by
\begin{subequations}
\begin{align}
    \tag{EEDMa}
    &\dot{x}_{i,1} = \sum_{j=1}^n x_{j,m} \mathcal{T}_{j,i}(\bar{x},p) - \lambda x_{i,1},\\ 
    \tag{EEDMb}
    &\dot{x}_{i,l} = \lambda (x_{i,l-1} - x_{i,l}),
\end{align} 
\end{subequations}
in which $p:=\mathcal{F}(\bar{x})$, $\bar{x} := \begin{bmatrix}\bar{x}_1 & \dots & \bar{x}_n\end{bmatrix}^T$ and
\begin{align*}
    \bar{x}_i := \sum_{l=1}^m x_{i,l}
\end{align*}
for all $i\in\{1,\dots,n\}$. We refer to $\bar{x}$ as the mean population state, $x$ as the extended mean population state, $p$ as the deterministic payoff and the dynamical system given by (EEDM) as the Erlang Evolutionary Dynamics Model (Erlang EDM). For all $t\geq 0$, we have $\bar{x}(t)\in\Delta$ and $x(t)\in\mathbb{X}$, where
\begin{align*}
\mathbb{X}:=\left\{\xi\in\mathbb{R}_{\geq 0}^{nm} \ \bigg | \ \sum_{i=1}^n\sum_{l=1}^m\xi_{i,l}=1\right\}.
\end{align*}
In the remainder of the paper, given $\xi\in\mathbb{R}^{nm}$, we denote $\xi=\begin{bmatrix} \xi_{1,1} & \dots & \xi_{1,m} & \dots & \xi_{n,1} & \dots & \xi_{n,m} \end{bmatrix}^T$ and $\bar{\xi}_i:=\sum_{l=1}^m\xi_{i,l}$ for any $i\in\{1,\dots,n\}$. 

Importantly, the discussions in~\cite[Section~V]{Park2018Payoff-Dynamic-} and~\cite[Appendix~12.B]{Sandholm2010Population-Game} indicate that, if a set $\mathbb{S}$ is globally attractive under the Erlang EDM, then the stationary distributions of $X$ concentrate near $\mathbb{S}$ as the number of agents tends to infinity. This result and the assumption that $N$ is large legitimizes the stability analysis carried out in the subsequent sections.

Note that, if $m=1$, then the Erlang EDM reduces to the standard EDM that results from the traditional framework \cite{Sandholm2010Population-Game,Park2018Payoff-Dynamic-}. This agrees with the fact that, when $m=1$, the constructions of the revision times in \S\ref{subsubsec:Trad_Rev_Times} and \S\ref{subsec:Erl_Rev_Times} coincide. 
Furthermore, the Erlang EDM conforms to the higher order evolutionary dynamics format, which requires the number of states to be greater than the number of strategies. Instances of such dynamics have been analyzed in \cite{Higher_order_game_dynamics_Laraki_Mertikopoulos,Passivity_analysis_of_higher_order_evolutionary_dynamics_and_population_games_Mabrok_Shamma}, although the results therein do not address the dynamics that we investigate in this paper.


\subsection{Prelude to Stability Analysis}

In the upcoming sections, we analyze the stability properties of the Erlang EDM. However, to have a meaningful analysis, we need further structure on the revision protocols and the game.

\subsubsection{Pairwise Comparison Protocols}
\label{subsubsec:Pair_Comp_Prots}

An important class of protocols that induce preferable stability results is the pairwise comparison class \cite{Sandholm2010Pairwise-compar}.

\begin{definition}
A protocol $\mathcal{T}$ is said to belong to the pairwise comparison (PC) class if for all $i\in\{1,\dots,n\}$, $j\in\{1,\dots,n\}\setminus\{i\}$ and $\xi,\pi\in\mathbb{R}^n$ it can be written as
\begin{align*}
    \mathcal{T}_{i,j}(\xi,\pi)=\phi_{i,j}(\pi),
\end{align*} 
where $\phi_{i,j}:\mathbb{R}^n\to\mathbb{R}_{\geq 0}$ satisfies sign preservation in the sense that $\phi_{i,j}(\pi)>0$ if $\pi_j>\pi_i$ and $\phi_{i,j}(\pi)=0$ if $\pi_j\leq\pi_i$.
\end{definition}

Essentially, an agent following a PC protocol can only switch to strategies with payoffs that are greater than the payoff of its current strategy. 
Their desirable incentive properties \cite[\S2.5]{Sandholm2010Pairwise-compar} and inherently fully decentralized operation result in the applicability of the PC class in many engineering problems. For instance, the Smith protocol \cite{Smith1984The-stability-o}, which belongs to the PC class, has been widely used to study traffic problems.

Thus, in the remainder of this paper, we consider the Erlang EDM under the assumption that $\mathcal{T}$ is a PC protocol. We refer to the resulting dynamics as the Erlang Pairwise Comparison EDM (Erlang PC-EDM).

Confining the protocol to be of the PC class readily yields a desirable characteristic. Namely, leveraging the so-called Nash stationarity of PC protocols \cite{Sandholm2010Pairwise-compar}, we identify the equilibria of the Erlang PC-EDM as
\begin{equation*}
\mathbb{ENE}(\mathcal{F}):= \left \{\xi\in\mathbb{X}\ \big | \ \bar{\xi}\in\mathbb{NE}(\mathcal{F}), \ \xi_{i,l}=\tfrac{1}{m}\bar{\xi}_i \right \},
\end{equation*}
which has a Nash-like form in the sense that $\bar{\xi}\in\mathbb{NE}(\mathcal{F})$ for all $\xi\in\mathbb{ENE}(\mathcal{F})$.

\subsubsection{Potential and Contractive Game}
\label{subsubsec:Pot_and_Con_Games}

Having confined our attention to the Erlang PC-EDM, we ask whether the mean population state converges to $\mathbb{NE}(\mathcal{F})$, i.e., whether
\begin{align*}
    \lim_{t\to\infty}\inf_{\bar{\xi}\in\mathbb{NE}(\mathcal{F})}\|\bar{x}(t)-\bar{\xi}\|=0.
\end{align*}
For the answer of this question to be affirmative, we need to assume some structure also on the game.\footnote{Assumptions on the game is necessary to guarantee stable behavior under PC protocols even in the traditional framework \cite[Chapter~9]{Sandholm2010Population-Game}.}

Two classes of games that introduce such structure are potential \cite{Monderer1996Potential-games,Sandholm2001Potential-games} and strictly contractive games \cite{Hofbauer2009Stable-games-an}.

\begin{definition} \label{def:pot_game}
A game $\mathcal{F}$ is said to be a potential game if there is a continuously differentiable function $f:\mathbb{R}^n\to\mathbb{R}$ satisfying $\nabla f = \mathcal{F}$. We refer to $f$ as the game's potential.
\end{definition}

\begin{definition} \label{def:con_game}
A game $\mathcal{F}$ is said to be contractive if $\eta^TD\mathcal{F}(\xi)\eta\leq 0$ for all $\eta\in T\Delta:=\{\nu\in\mathbb{R}^n \ | \ \sum_{i=1}^n\nu_i=0\}$ and $\xi\in\Delta$, where $D\mathcal{F}$ denotes the Jacobian of $\mathcal{F}$. Moreover, $\mathcal{F}$ is said to be strictly contractive if $\eta^TD\mathcal{F}(\xi)\eta< 0$, in which case we define 
\begin{align*}
    &\ubar{\gamma} := -\max_{\xi\in\Delta,\eta\in T\Delta}\eta^T D\mathcal{F}(\xi)\eta,\\
    &\bar{\gamma} := -\min_{\xi\in\Delta,\eta\in T\Delta}\eta^T D\mathcal{F}(\xi)\eta.
\end{align*}
\end{definition}

We note that the class of potential and strictly contractive games do not contain one another. For instance, the 123-coordination game \cite[Example~3.1.5]{Sandholm2010Population-Game} is potential, but not contractive, and the ``good'' rock-paper-scissors (RPS) game \cite[Example~3.3.2]{Sandholm2010Population-Game} is strictly contractive, but not potential.

\section{Results Under Potential Games}
\label{sec:Res_Und_Pot_Games}

In this section, we assume that $\mathcal{F}$ is potential and show that the mean population state converges to $\mathbb{NE}(\mathcal{F})$. Our analysis follows a similar approach to that in \cite{Sandholm2001Potential-games}, which proposes to investigate the potential of the game evaluated along the trajectories of the mean population state.

\begin{theorem} \label{thm:pot}
If 	$\mathcal{F}$ is a potential game, then
\begin{align*}
\lim_{t\to\infty}\inf_{\bar{\xi}\in \mathbb{NE}(\mathcal{F})} \|\bar{x}(t)-\bar{\xi}\|=0.
\end{align*}
\end{theorem}

\begin{proof}
Since $\mathcal{F}$ is a potential game, it has a potential $f$ as specified in Definition~\ref{def:pot_game}. Let us define $\mathcal{L}:\mathbb{R}^n\to\mathbb{R}$ as $\mathcal{L}(\xi) = -f(\bar{\xi})$. Taking the time-derivative of $\mathcal{L}$ along the trajectories of the Erlang PC-EDM yields
\begin{align}
    &-\frac{d}{dt}f(\bar{x}) = -\sum_{i=1}^n\sum_{j=1}^n x_{i,m}\phi_{i,j}(p)(p_j-p_i) \leq 0, \label{ineq:dfdt}
\end{align}
where the inequality in \eqref{ineq:dfdt} follows from the sign-preservation property of PC protocols. Moreover, the inequality in \eqref{ineq:dfdt} holds with equality if and only if, whenever $i,j\in\{1,\dots,n\}$ and $t\geq 0$ satisfies $p_j(t)>p_i(t)$, we have $x_{i,m}(t)=0$. Noting that $\mathbb{X}$ is compact and positively invariant under the Erlang PC-EDM (see \S\ref{subsubsec:Erl_Evol_Dyn}), it follows from LaSalle's invariance principle \cite[Theorem~3.4]{Nonlinear_Systems_Khalil} that $x$ converges to the largest invariant subset of $\Theta$ defined below:
\begin{equation*}
\Theta:=\left \{\xi\in\mathbb{X} \ \bigg | \ \xi_{i,m}>0 \implies  i\in\argmax_{j\in\{1,\dots,n\}} \mathcal{F}_j(\bar{\xi}) \right \},
\end{equation*} where $i$ ranges from $1$ to $n$.
Subsequently, we show that such a largest invariant subset of $\Theta$ is in $\{\xi\in\mathbb{X} \ | \ \bar{\xi}\in\mathbb{NE}(\mathcal{F})\}$.

To begin with, notice that $\Theta\supseteq\{\xi\in\mathbb{X} \ | \ \bar{\xi}\in\mathbb{NE}(\mathcal{F})\}$. Furthermore, observe that for all $t\geq 0$ satisfying $x(t)\in\Theta$, we have $\dot{\bar{x}}(t) = 0$. This implies that $\{\xi\in\mathbb{X} \ | \ \bar{\xi}\in\mathbb{NE}(\mathcal{F})\}$ is invariant under the Erlang PC-EDM. 

Now, take $\xi\in\Theta$ such that $\bar{\xi}\notin \mathbb{NE}(\mathcal{F})$ and consider the trajectory of the Erlang PC-EDM from the initial state $\xi$. Since $\bar{\xi}\notin \mathbb{NE}(\mathcal{F})$, the set $\mathcal{I}=\{i\notin \arg\max_{j\in\{1,\dots,n\}} \mathcal{F}_j(\bar{\xi}) \ | \ \bar{\xi}_i>0\}$ is non-empty. Therefore, from (EEDMb) and the fact that $p$ is stationary at states in $\Theta$, it follows that there exist some $i\in\mathcal{I}$ and $\ubar{t}\geq 0$ such that $\dot{x}_{i,m}(\ubar{t})>0$ and $\dot{\bar{x}}(t)=0$ for all $t\in[0,\ubar{t}]$. Consequently, the trajectory leaves $\Theta$ at time $\ubar{t}$. As a result, the largest invariant subset of $\Theta$ under the Erlang PC-EDM has to be a subset of $\{\xi\in\mathbb{X}\ | \ \bar{\xi}\in\mathbb{NE}(\mathcal{F})\}$.
\end{proof}

When the game $\mathcal{F}$ has a strictly concave potential, Theorem~\ref{thm:pot} can be augmented to arrive at the following corollary.

\begin{corollary} 
\label{cor:pot}
If $\mathcal{F}$ is a potential game with a strictly concave potential, then
\begin{align*}
\lim_{t\to\infty}\inf_{\xi\in\mathbb{ENE}(\mathcal{F})}\|x(t)-\xi\|=0.
\end{align*}
\end{corollary}

\begin{proof}
If $\mathcal{F}$ has a strictly concave potential $f$, then $\mathbb{NE}(\mathcal{F})=\{\bar{\xi}^*\}$, where $\bar{\xi}^*$ is the unique maximizer of $f$ over $\Delta$ \cite[Corollary~3.1.4]{Sandholm2010Population-Game}. From Theorem~\ref{thm:pot}, it follows that $\lim_{t\to\infty}\bar{x}(t)=\bar{\xi}^*$. Moreover, $\dot{\bar{x}}$ is uniformly continuous, because $\mathcal{T}$ and $\mathcal{F}$ are Lipschitz continuous and $x$ takes values in a compact set. Hence, leveraging Barbalat's lemma, we obtain $\lim_{t\to\infty}\dot{\bar{x}}(t)= 0$.

Now, consider the dynamics \eqref{eq:tilde_x_state} of the auxiliary state $\tilde{x}$ defined in Appendix~\ref{app:Aux_States}. From $A$ (which is given in \eqref{eq:AB}) being Hurwitz and $\lim_{t\to\infty}\dot{\bar{x}}(t)= 0$, we have $\lim_{t\to\infty}\tilde{x}(t)=0$. Thus, for all $i\in\{1,\dots,n\}$ and $l\in\{1,\dots,m\}$, $\lim_{t\to\infty}|x_{i,l}(t)-x_{i,m}(t)|=0$. This and $\lim_{t\to\infty}\bar{x}(t)=\bar{\xi}^*$ imply that $\lim_{t\to\infty}\inf_{\xi\in\mathbb{ENE}(\mathcal{F})}\|x(t)-\xi\|=0.$
\end{proof}

\section{Convergence For Strictly Contractive Games}
\label{sec:Res_Und_Str_Con_Games}

In this section, we assume that $\mathcal{F}$ is strictly contractive and present a condition that ensures the convergence of the extended mean population state to $\mathbb{ENE}(\mathcal{F})$.

When the game is strictly contractive, it is known that the standard PC-EDM does not necessarily exhibit stable behavior (see \cite[Exercise~7.2.10]{Sandholm2010Population-Game}). Since the Erlang PC-EDM with $m=1$ corresponds to the standard PC-EDM \cite{Sandholm2010Pairwise-compar}, these instability results are inherited by the Erlang PC-EDM. Nonetheless, when the game is strictly contractive, stability of the standard PC-EDM can be ensured by assuming that the protocol belongs to a refinement of the PC class. This refinement requires the PC protocol to be impartial according to the definition given below.

\begin{definition} \label{def:imp}
A protocol $\mathcal{T}$ is said to be of the impartial pairwise comparison (IPC) class if for $i,j$ in $\{1,\dots,n\}$, with $i\neq j$, and $\pi\in\mathbb{R}^n$ it can be written as $\mathcal{T}_{i,j}(\xi,\pi)=\phi_j(\pi_j-\pi_i)$, where $\phi_j:\mathbb{R}\to\mathbb{R}_{\geq 0}$ is sign preserving.
\end{definition}

Consequently, in our analysis of the Erlang PC-EDM regarding strictly contractive games, we consider IPC protocols.

In addition to impartiality, given a strictly contractive $\mathcal{F}$, we use the following constant to establish the global attractivity of $\mathbb{ENE}(\mathcal{F})$ under the Erlang PC-EDM:
\begin{align}
    \ubar{\lambda} := 2c\bar{\sigma}\left( \frac{n\bar{\gamma}}{(m+1)\ubar{\gamma}}\right)^{1/2}. \label{eq:lambda_bd}
\end{align}
Here, $\bar{\gamma},\ubar{\gamma}$ are given in Definition~\ref{def:con_game}, $c$ is specified by \eqref{eq:c}, and $\bar{\sigma}:=\sup_{\omega\in[0,\infty)}\sigma_{\max}((j\omega-A)^{-1}B)$, where $\sigma_{\max}$ denotes the maximum singular value, and $A$, $B$ are given by
\begin{align}
\label{eq:AB}
    A := \TinyMatrix{
    -1 & 0 & \dots & 0 & -1\\
    1 & -1 & \dots & 0 & -1\\
    0 & 1 & \dots & 0 & -1\\
    \vdots & \vdots & \ddots & \vdots & \vdots\\
    0 & 0 & \dots & 1 & -2
	}\otimes \textbf{I}_n,~B := \textbf{e}_1\otimes \textbf{I}_n,
\end{align}
in which $\otimes$ denotes the Kronecker product, $\textbf{I}_n$ is the $n\times n$ identity matrix and $\textbf{e}_1$ is the first standard basis vector in $\mathbb{R}^{m-1}$. We note that $\bar{\sigma}$ is the $H_{\infty}$-norm of the linear system specified by $\dot{z}=Az+Bu$ (with input $u$ and output $z$). For the case when $m\leq 4$, we can compute $\bar{\sigma}$ simply as $\bar{\sigma}=((2m^2-3m+1)/(6m))^{1/2}$. As for the $m>4$ case, computation of $\bar{\sigma}$ is more challenging, yet can be done numerically via the bisection $H_\infty$-norm computation
algorithm \cite{A_bisection_method_for_computing_the_H_infty_norm_of_a_transfer_matrix_and_related_problems_Boyd_et_al}.

Having presented IPC protocols and introduced the constant $\ubar{\lambda}$, we are now ready to state the theorem which provides a condition that ensures the global attractivity of $\mathbb{ENE}(\mathcal{F})$ under the Erlang PC-EDM.

\begin{theorem} \label{thm:str_con}
If $\mathcal{F}$ is strictly contractive, the protocol is impartial and $\lambda>\ubar{\lambda}$, then 
\begin{align*}
\lim_{t\to\infty}\inf_{\xi\in\mathbb{ENE}(\mathcal{F})}\|x(t)-\xi\|=0.
\end{align*}
\end{theorem}

We present a proof of Theorem~\ref{thm:str_con} in Appendix~\ref{app:Pf_of_Str_Con}, which follows mainly from the two time-scale structure of the Erlang EDM. Namely, when $\lambda$ is large in comparison to $c$, the dynamics associated with the sub-strategies gives the ``fast'' part of (EEDM), whereas the dynamics of $\bar{x}$ gives its ``slow'' part. Thus, for any $i\in\{1,\dots,n\}$, $x_{i,1},\dots,x_{i,m}$ rapidly equalize and closely track $\bar{x}_i/m$. Thereafter, (EEDM) approximates the standard EDM, and global attractivity of $\mathbb{ENE}(\mathcal{F})$ ensues from the stability properties of the standard PC-EDM \cite{Hofbauer2007Stable-games,Park2018Payoff-Dynamic-}.

Note from Theorem~\ref{thm:str_con} and the definition of $\mathbb{ENE}(\mathcal{F})$ that, if $\lambda>\ubar{\lambda}$, the game $\mathcal{F}$ is strictly contractive and the PC protocol is impartial, then the mean population state converges to $\mathbb{NE}(\mathcal{F})$.

\section{Numerical Examples}
\label{sec:Num_Ex}

We proceed to illustrate our results via two population games. In both examples, we assume that the revision protocol is the Smith protocol \cite{Smith1984The-stability-o}, meaning that the $\mathcal{T}_{i,j}(\xi,\pi)$ in (EEDMa) is given for all $\bar{\xi}\in\Delta$, $\pi\in\mathbb{R}^n$ and $i,j\in\{1,\dots,n\}$ such that $i\neq j$ by $\mathcal{T}_{i,j}(\xi,\pi) = \max(\pi_j-\pi_i,0)$.

\subsection{A Congestion Game Example}
\label{subsec:Con_Game_Ex}

A well-studied example of population games is congestion games \cite[Chapter~2.2.2]{Sandholm2010Population-Game}. We utilize the congestion game characterized by the graph in Figure~\ref{fig:con_game} to demonstrate how Theorem~\ref{thm:pot} can come into play.

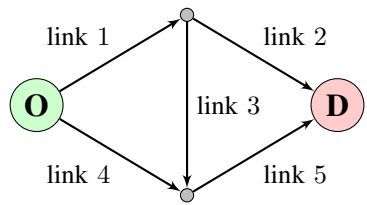
\begin{figure}[ht]
\centering
\begin{tikzpicture}[auto, node distance=7em,>=latex', scale = 0.8]
\node [circle,draw=black, fill=green!20, inner sep=0pt,minimum size=20pt] (origin) at (0,0) {\bf \large O};
\node [circle,draw=black, fill=red!20, inner sep=0pt,minimum size=20pt] (destination) at (5,0) {\bf \large D};
\node [circle,draw=black, fill=black!25, inner sep=0pt,minimum size=5pt] (vertex1) at (2.5,1.5) {};
\node [circle,draw=black, fill=black!25, inner sep=0pt,minimum size=5pt] (vertex2) at (2.5,-1.5) {};
\draw[thick,color=black, ->] (origin) --  (vertex1) node [pos=.5] {link $1$};
\draw[thick,color=black, ->] (vertex1) --  (destination) node [pos=.5] {link $2$};
\draw[thick,color=black, ->] (vertex1) --  (vertex2) node [pos=.5] {link $3$};
\draw[thick,color=black, ->] (origin) --  (vertex2) node [pos=.5, below left] {link $4$};
\draw[thick,color=black, ->] (vertex2) --  (destination) node [pos=.5, below right] {link $5$};
\end{tikzpicture}
\caption{Congestion game example with one origin/destination pair and 3 strategies.}
\label{fig:con_game}
\end{figure}

In this graph, O denotes the origin, D denotes the destination, links represent roads and arrows on links represent the direction in which an agent choosing the link travels. Agents can choose to go from the origin to the destination via one of the three available routes. Accordingly, these routes constitute the strategies available to the agents. To each link $l\in\{1,\dots,5\}$, we assign a utilization-dependent cost given by $c_l\sum_{\{i \ | \ i\in\Omega_l\}}\bar{\xi}_i$, where $\Omega_l$ denotes the set of routes in which link $l$ is used, $\bar{\xi}_i$ is the percentage of agents playing strategy $i$ and $c_l$ is a positive constant quantifying how well the link accommodates traffic. Hence, the payoffs of using the routes under the population state value $\bar{\xi}\in\Delta$ is
\begin{align*}
    \mathcal{F}^{Con}(\bar{\xi})=-
    \TinyMatrix{
    c_1+c_2 & 0 & c_1\\
    0 & c_4+c_5 & c_5\\
    c_1 & c_5 & c_1+c_3+c_5\\}
    \bar{\xi},
\end{align*}
where the route formed by links 1, 2 is defined to be the first strategy, links 4, 5 form the second strategy, and the third strategy is the route given by links 1, 3, 5. Suppose that the parameters of this congestion game is given by $c_1 =2.5$, $c_2 = 1.5$, $c_3 = 0.5$, $c_4 = 2.5$, and $c_5 = 0.7$.

Let us consider the Erlang PC-EDM induced by $\mathcal{F}^{Con}$ and the Smith revision protocol. Moreover, assume that the number of sub-strategies is $m=3$ and the revision rate is $\lambda=5$. Since congestion games are potential games \cite[Example~3.1.2]{Sandholm2010Population-Game} and the Smith protocol belongs to the PC class, we can invoke Theorem~\ref{thm:pot} to conclude that the mean population state resulting from the specified Erlang PC-EDM converges to $\mathbb{NE}(\mathcal{F}^{Con})$.

To display this outcome, we simulate the Erlang PC-EDM with the aforementioned specifications from the initial state given by $x_{1,3}(0)=x_{2,1}(0)=0.2$, $x_{3,1}(0)=0.6$ and $x_{i,l}(0)=0$ for all other $i,l\in\{1,2,3\}$. The trajectory of the mean population state acquired from this simulation is portrayed in Fig.~\ref{fig:sim_game_pot}. Observe from Fig.~\ref{fig:sim_game_pot} that the mean population state converges to the Nash equilibrium of $\mathcal{F}^{Con}$, which is the singleton at approximately $(0.349,0.513,0.137)$. Fig.~\ref{fig:sim_game_pot} also portrays the trajectory of the mean population state acquired by simulating the standard counterpart of the Erlang PC-EDM with the aforementioned specifications from the same initial state. Recall from \S\ref{subsubsec:Erl_Evol_Dyn} that the standard counterpart conforms to the traditional framework (see \S\ref{subsubsec:Trad_Rev_Times}) and is obtained by setting $m=1$ while keeping all the other parameters of the Erlang PC-EDM unchanged.

\begin{figure}[ht]
\centering
\begin{tikzpicture}[scale = 0.7]
\begin{axis}[view = {-45}{-45},
opacity=0.55,
clip=false,
3d box=complete,
grid,
grid style={dashed,gray!40},
axis line style={gray!40},
legend style={at={(axis cs:0.3,0.15,0.55)}, nodes={scale=1.3, transform shape}, opacity=1}
] 
\addplot3[thick, blue, opacity=1] table [col sep=comma, mark=none] {data/con_game_data1_3d_v2.txt};
\addplot3[thick, red, opacity=1] table [col sep=comma, mark=none] {data/con_game_data2_3d_v2.txt};
\addlegendentry{Erlang}
\addlegendentry{Standard}
\draw [black, fill=black, opacity=1] (0.2,0.2,0.6) circle (2pt);
\draw (0.2,0.2,0.6) node[font=\large,opacity=1,anchor=south west]{(0.2,0.2,0.6)};
\draw [black, fill=black, opacity=1] (0.349,0.513,0.137) circle (2pt);
\draw (0.369,0.48,-0.087) node[font=\large,opacity=1,anchor=south west]{(0.349,0.513,0.137)};
\end{axis}
\end{tikzpicture}
\caption{Mean population states from the Erlang and standard PC-EDMs induced by $\mathcal{F}^{Con}$ and the Smith protocol.}
\label{fig:sim_game_pot}
\end{figure}
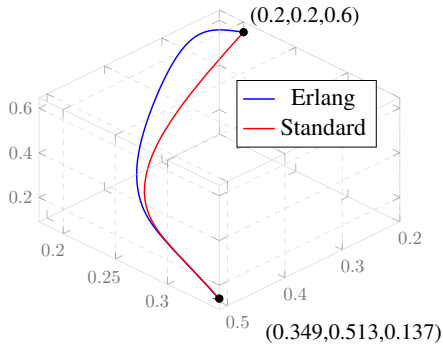

\subsection{A Rock-Paper-Scissors Game Example}

As noted in \S\ref{subsubsec:Pot_and_Con_Games}, the class of potential games and strictly contractive games do not contain one another, and the good RPS game \cite[Example~3.3.2]{Sandholm2010Population-Game} is an example of a strictly contractive game that is not potential.

Let us specify the good RPS game by
\begin{align*}
    \mathcal{F}^{RPS}(\bar{\xi})=\TinyMatrix{0 & -2 & 3\\ 3 & 0 & -2 \\ -2 & 3 & 0}\bar{\xi}
\end{align*}
for all $\bar{\xi}\in\Delta$, and consider the Erlang PC-EDM induced by the Smith protocol and $\mathcal{F}^{RPS}$. Moreover, let the number of sub-strategies be $m=4$. Since $\mathcal{F}^{RPS}$ is strictly contractive but not potential, Theorem~\ref{thm:pot} can't be utilized and we have to resort to Theorem~\ref{thm:str_con}. To apply Theorem~\ref{thm:str_con}, we compute $(\bar{\gamma},\ubar{\gamma},c)$ as $(1,1,4)$. Furthermore, as stated in \S\ref{sec:Res_Und_Str_Con_Games}, when $m\leq 4$ we have $\bar{\sigma}=((2m^2-3m-1)/(6m))^{1/2}$, meaning that for $m=4$ the value of $\bar{\sigma}$ is $0.9354$. Hence, we obtain $\ubar{\lambda}=5.7965$. Noting that the Smith protocol belongs to the IPC class, it follows from Theorem~\ref{thm:str_con} that, if $\lambda>5.7965$, then the extended mean population state resulting from the specified Erlang PC-EDM converges to $\mathbb{ENE}(\mathcal{F}^{RPS})$.

To display this outcome, we set $\lambda=5.8$ and simulate the Erlang PC-EDM with the aforementioned specifications from the initial state given by $x_{1,4}(0)=x_{2,1}(0)=0.2$, $x_{3,1}(0)=0.6$ and $x_{i,l}(0)=0$ for all other $i\in\{1,2,3\}$, $l\in\{1,2,3,4\}$. The trajectory of the mean population state acquired from this simulation, along with the trajectory from the simulation of its standard counterpart, is portrayed in Fig.~\ref{fig:sim_game_con}. We can observe from Fig.~\ref{fig:sim_game_con} that the mean population state converges to the Nash equilibrium of $\mathcal{F}^{RPS}$, which is the singleton at $(1/3,1/3,1/3)$.

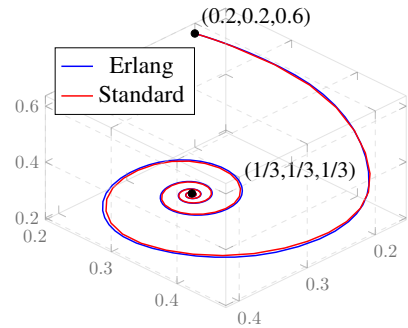
\begin{figure}[ht]
\centering
\begin{tikzpicture}[scale = 0.7]
\begin{axis}[view = {-45}{-45},
opacity=0.55,
clip=false,
3d box=complete,
grid,
grid style={dashed,gray!40},
axis line style={gray!40},
legend style={at={(axis cs:0.2,0.2,0.55)}, nodes={scale=1.3, transform shape}, opacity=1}
] 
\addplot3[thick, blue, opacity=1] table [col sep=comma, mark=none] {data/rps_data1_3d_v1.txt};
\addplot3[thick, red, opacity=1] table [col sep=comma, mark=none] {data/rps_data2_3d_v1.txt};
\addlegendentry{Erlang}
\addlegendentry{Standard}
\draw [black, fill=black, opacity=1] (0.2,0.2,0.6) circle (2pt);
\draw (0.2,0.2,0.6) node[font=\large,opacity=1,anchor=south west]{(0.2,0.2,0.6)};
\draw [black, fill=black, opacity=1] (0.33,0.33,0.33) circle (2pt);
\draw (0.4,0.33,0.43) node[font=\large,opacity=1,anchor=south west]{(1/3,1/3,1/3)};
\end{axis}
\end{tikzpicture}
\caption{Mean population states from the Erlang and standard PC-EDMs induced by $\mathcal{F}^{RPS}$ and the Smith protocol.}
\label{fig:sim_game_con}
\end{figure}

\section{Conclusion and Future Directions}
\label{sec:Conc}

In this paper, we present an extension of the population games and evolutionary dynamics paradigm by allowing agents' inter-revision times to be i.i.d. Erlang random variables. We show that the long term behavior of the population state resulting from this generalization can be inferred by analyzing, what we call, the Erlang EDM. Then, we confine our focus to PC revision protocols and consider the Erlang PC-EDM. When the game is potential, we show that the mean population state converges to $\mathbb{NE}(\mathcal{F})$ for any revision rate and number of sub-strategies. Similarly, when the game is strictly contractive, we show that $\mathbb{ENE}(\mathcal{F})$ is globally attractive under the Erlang PC-EDM provided that the protocol is impartial and $\lambda$ satisfies a bound condition.

The work presented in this paper also raises questions for future research. For instance, despite the results in \S\ref{sec:Res_Und_Str_Con_Games}, it is still unclear whether global attractivity of $\mathbb{ENE}(\mathcal{F})$ under the Erlang PC-EDM induced by an impartial protocol and strictly contractive game is guaranteed for any revision rate. Moreover, \cite{Fox2013Population-Game,Park2018Payoff-Dynamic-,Arcak2020Dissipativity-T} generalizes the class of admissible payoff mechanisms to so-called payoff dynamics models; however we only consider static games. Hence, it can be investigated whether the analysis in \S\ref{sec:Res_Und_Str_Con_Games} can be altered to fit the $\delta$-passivity \cite{Park2018Payoff-Dynamic-} or $\delta$-dissipativity \cite{Arcak2020Dissipativity-T} framework, which would broaden the global attractivity results therein to accommodate more general payoff structures.

\bibliographystyle{IEEEtran}
\bibliography{MartinsRefs,KaraRefs}

\begin{appendix}

\subsection{Auxiliary Notation and Analysis}
\label{app:Aux_States}

In this section we introduce the auxiliary states used in the proofs of Corollary~\ref{cor:pot}~and~Theorem~\ref{thm:str_con}, and characterize the dynamics of these states.

Given any $l\in\{1,\dots,m\}$, let $E_l$ be the $n\times nm$ matrix
\begin{align*}
E_l := \textbf{I}_n \otimes \textbf{e}_l^T,
\end{align*}
where $\textbf{e}_l$ is the $l$-th standard basis vector in $\mathbb{R}^m$. Moreover, given $\xi\in\mathbb{R}^{nm}$ we denote
\begin{align*}
    \tilde{\xi} := \begin{bmatrix} (E_1-E_m)\xi & \dots & (E_{m-1}-E_m)\xi \end{bmatrix}^T.
\end{align*}
Now, let us introduce the auxiliary states $y$, $\delta$ and $\tilde{x}$:
\begin{align*}
     y &:= \begin{bmatrix} y_{1,1} & \dots & y_{1,m-1} & \dots & y_{n,1} & \dots & y_{n,m-1}\end{bmatrix}^T, \\
     \delta &:= \begin{bmatrix} \delta_1^T & \dots & \delta_{m-1}^T
    \end{bmatrix}^T, \quad
    \tilde{x}:=\begin{bmatrix} \tilde{x}_1 & \dots & \tilde{x}_{m-1}\end{bmatrix}^T,
\end{align*}
where, for any $i\in\{1,\dots,n\}$ and $l\in\{1,\dots,m-1\}$ we set $\delta_l := \begin{bmatrix} y_{1,l} & \dots & y_{n,l}\end{bmatrix}^T$, $\tilde{x}_l := (E_l-E_m)x$, and define $y_{i,l}$ to be a solution of $\dot{y}_{i,l} = \sum_{j=1}^n \phi_i(p_i-p_j)x_{j,l}-\lambda x_{i,l}$.

Notice that, from the definition of $\tilde{x}$ we have $\dot{\tilde{x}}_l = (E_l-E_m)\dot{x}$. Consequently,
\begin{align}
    \dot{\tilde{x}} &= \lambda A\tilde{x}  + B\dot{\bar{x}}, \label{eq:tilde_x_state}
\end{align}
where we remind that $A$ and $B$ are given in \eqref{eq:AB}.

Moreover, let us define $\Phi:\mathbb{R}^n\to\mathbb{R}^{n\times n}$ as the matrix-valued function given for all $\pi\in\mathbb{R}^n$ by
\begin{align*}
    \Phi_{ij}(\pi) = 
    \begin{cases}
    \phi_i(\pi_i-\pi_j),\quad\text{if}~i\neq j,\\
    \sum_{j=1}^n\phi_j(\pi_j-\pi_i),\quad\text{if}~i= j.\\
    \end{cases}
\end{align*}
Then, for all $l\in\{1,\dots,m-1\}$, we can write
\begin{equation} \label{eq:del_state}
    \dot{\delta}_l = \Phi(p)E_lx - \Phi(p)E_mx 
    = \Phi(p)\tilde{x}_l. 
\end{equation}

\subsection{Proof of Theorem~\ref{thm:str_con}}
\label{app:Pf_of_Str_Con}

We proceed to present a proof of Theorem~\ref{thm:str_con} and the discussion that leads up to it. Our approach is based on analyzing the function $\mathcal{L}_{\alpha}:\mathbb{R}^n\times\mathbb{R}^n\to\mathbb{R}_{\geq 0}$ inspired by the Lyapunov function for the standard IPC-EDM \cite{Hofbauer2009Stable-games-an}. Namely, based on a modification of that in~\cite[Theorem~7.1]{Hofbauer2009Stable-games-an}, we set
\begin{equation}
    \mathcal{L}_{\alpha}(\xi,\pi) := \sum_{i=1}^n \bar{\xi}_i \sum_{j=1}^n \Psi_j(\pi_j-\pi_i)+\alpha \tilde{\xi}^TM\tilde{\xi},
\end{equation}
where $M$ is the solution of the Lyapunov equation $A^TM+MA=-I$ (since $A$ is Hurwitz, such $M$ exists and is symmetric positive-definite), $\alpha$ is a positive constant satisfying $\alpha<(m+1)\ubar{\gamma}/(2\|MB\|_2^2)$, and $\Psi:\mathbb{R}^n\to\mathbb{R}^n$ is given by
\begin{align*}
       &\Psi_j(\pi_j-\pi_i) :=\int_0^{\pi_j-\pi_i} \phi_j(s)ds.
\end{align*}

From a procedure similar to that in \cite[Appendix~A.4]{Hofbauer2009Stable-games-an}, we obtain the following time-derivative of $\mathcal{L}_\alpha$ along the trajectories of a solution of the Erlang PC-EDM and the deterministic payoff:
\begin{equation} \label{eq:dLdt}
    \tfrac{d}{dt} \mathcal{L}_{\alpha}(x,p) = -\mathcal{P}(x,p)+\mathcal{Q}(x,p),
\end{equation} 
where $\mathcal{P}(x,p)$ and $\mathcal{Q}(x,p)$ are specified as
\begin{align*}
    &\mathcal{P}(x,p) = -\alpha\lambda\tilde{x}^T\tilde{x}\\
    &\quad +\sum_{i,j=1}^n \phi_i(p_i-p_j) x_{j,m} \sum_{k=1}^n \Psi_k(p_k-p_i)-\Psi_k(p_k-p_j),\\
    &\mathcal{Q}(x,p) = \Big( m\dot{\bar{x}}^T\dot{p} + \begin{bmatrix}\dot{p}^T & \dots & \dot{p}^T\end{bmatrix}\dot{\delta}+ \alpha2\tilde{x}^TMB\dot{\bar{x}}\Big).
\end{align*}

The argument in \cite[Appendix~A.4]{Hofbauer2009Stable-games-an} can be readily adapted to prove the following proposition.

\begin{proposition} \label{prop:P}
If $\mathcal{F}$ is contractive, then for all $\xi\in\mathbb{X}$ and $\pi\in\mathbb{R}^n$ we have $\mathcal{P}(\xi,\pi) \geq 0$.
\end{proposition}

We now focus on the $\mathcal{Q}$ term. The following proposition is a key step in proving Theorem~\ref{thm:str_con}.

\begin{proposition}  \label{prop:Q}
Assume that $\mathcal{F}$ is strictly contractive and $\lambda$ satisfies
\begin{align}
\lambda \geq \left(\frac{2\bar{\sigma}(\alpha+2\bar{\gamma}nc^2)}{(m+1)\ubar{\gamma}-2\alpha \|MB\|_2^2}\right)^{1/2}, \label{ineq:pass_param}
\end{align}
where $\bar{\gamma},\ubar{\gamma}$ are given in Definition~\ref{def:con_game}, $c$ is specified by \eqref{eq:c} and $\bar{\sigma}$ is the supremum of the maximum singular value of $((j\omega-A)^{-1}B)$ over $\omega\in[0,\infty)$. Then, the following holds for all $t\geq 0$:
\begin{align}
&\int_0^t \mathcal{Q}(x(\tau),p(\tau))d\tau \leq \left(\alpha+2\bar{\gamma}nc^2\right)\|e^{\lambda A t}\tilde{x}(0)\|_2^2. \label{ineq:Q}
\end{align}
\end{proposition}

\begin{proof}
We begin by deriving a bound on $\int_0^t\|\tilde{x}(\tau)\|^2_2d\tau$. Observe from \eqref{eq:tilde_x_state} that
\begin{align*}
    \tilde{x}(t) = e^{\lambda A t}\tilde{x}(0)+\int_0^t e^{\lambda A (t-\tau)}B\dot{\bar{x}}(\tau)d\tau.
\end{align*}
Thus, utilizing Parseval's theorem, we get
\begin{align}
    \int_{0}^{t}\|\tilde{x}(\tau)\|^2_2d\tau \leq \|e^{-\lambda A t}\tilde{x}(0)\|_2^2+\frac{\bar{\sigma}}{\lambda^2} \int_0^{t}\|\dot{\bar{x}}(\tau)\|^2_2d\tau. \label{ineq:dotxtilde}
\end{align}

We proceed by deriving a bound on $\|\dot{\delta}(t)\|_2$. Notice from \eqref{eq:del_state} that for all $l\in\{1,\dots,m-1\}$ we have
\begin{align}
    \|\dot{\delta}_l(t)\|^2_2 &\leq \|\Phi(p(t))\|^2_2\|\tilde{x}_l(t)\|^2_2 \nonumber\\
    &\leq n\|\Phi(p(t))\|^2_1\|\tilde{x}_l(t)\|^2_2 = 4nc^2\|\tilde{x}_l(t)\|^2_2, \label{ineq:dotdelta}
\end{align}
where $c = \max_{\bar{\xi}\in\Delta}\sum_{j=1}^n\phi_j(\mathcal{F}_j(\bar{\xi})-\mathcal{F}_i(\bar{\xi}))$ exists, since $\phi$ and $\mathcal{F}$ are Lipschitz continuous, and $\Delta$ is compact.

Now, we leverage \eqref{ineq:dotxtilde} and \eqref{ineq:dotdelta} to obtain a condition that guarantees \eqref{ineq:Q}. Negative definiteness of $D\mathcal{F}(\bar{x})$ with respect to $T\Delta$ implies for all $l\in\{1,\dots,m-1\}$ that
\begin{align}
\frac{1}{2}(-\dot{\delta}_l^TD\mathcal{F}(\bar{x})\dot{\delta}_l-\dot{\bar{x}}^TD\mathcal{F}(\bar{x})\dot{\bar{x}})\leq |\dot{\delta}_l^TD\mathcal{F}(\bar{x})\dot{\bar{x}}| \label{ineq:ipbound}.
\end{align}
From \eqref{ineq:ipbound}, with $2|\tilde{x}^TMB\dot{\bar{x}}|\leq \|\tilde{x}\|_2^2+\|MB\|_2^2\|\dot{\bar{x}}\|_2^2$ and negative definiteness of $D\mathcal{F}(\bar{x})$ with respect to $T\Delta$, we get
\begin{align}
& \int_0^t m\dot{\bar{x}}(\tau)^T\dot{p}(\tau) + \sum_{l=1}^{m-1}\dot{\delta}_l(\tau)^T\dot{p}(\tau)+\alpha2\tilde{x}(\tau)^TMB\dot{\bar{x}}(\tau) d\tau \nonumber\\
&\leq \int_0^t -\frac{m+1}{2}\ubar{\gamma} \|\dot{\bar{x}}(\tau)\|_2^2 + \frac{1}{2}\bar{\gamma} \|\dot{\delta}(\tau)\|_2^2 \nonumber\\
&\qquad +\alpha\|\tilde{x}(\tau)\|^2_2+\alpha\|MB\|_2^2\|\dot{\bar{x}}(\tau)\|^2_2 d\tau. \label{ineq:pass_param_simp}
\end{align}
Finally, combining \eqref{ineq:pass_param_simp}, \eqref{ineq:dotxtilde} and \eqref{ineq:dotdelta}, it follows that
\begin{align}
& \int_0^t m\dot{\bar{x}}(\tau)^T\dot{p}(\tau) + \sum_{l=1}^{m-1}\dot{\delta}_l(\tau)^T\dot{p}(\tau)+\alpha2\tilde{x}(\tau)^TMB\dot{\bar{x}}(\tau) d\tau \nonumber\\
&\leq \left(\alpha+2\bar{\gamma}nc^2\right)\|e^{\lambda A t}\tilde{x}(0)\|_2^2 +\int_0^t \Bigg(-\frac{m+1}{2}\ubar{\gamma} \nonumber\\
&\quad +\alpha\|MB\|_2^2 + (\alpha + 2\bar{\gamma}nc^2)\frac{\bar{\sigma}}{\lambda^2} \Bigg)\|\dot{\bar{x}}(\tau)\|_2^2 d\tau.
\label{ineq:pass_param_simp_fin}
\end{align}
As a result, if \eqref{ineq:pass_param} holds, then
\begin{align}
\int_0^t \mathcal{Q}(x(\tau),p(\tau))d\tau \leq \left(\alpha+2\bar{\gamma}nc^2\right)\|e^{\lambda A t}\tilde{x}(0)\|_2^2.
\end{align}
\end{proof}

Now, we are ready to present a proof of Theorem~\ref{thm:str_con}, which is a direct consequence of Propositions~\ref{prop:P}~and~\ref{prop:Q}, and Barbalat's lemma.

\begin{proof}
Assume that $\lambda>\ubar{\lambda}$, where $\ubar{\lambda}$ is specified in \eqref{eq:lambda_bd}. Then, there exists $\alpha^*>0$ satisfying $\alpha^*<(m+1)\ubar{\gamma}/(2\|MB\|_2^2)$ such that \eqref{ineq:pass_param} holds with $\alpha=\alpha^*$. Thus, we can leverage Propositions~\ref{prop:P}~and~\ref{prop:Q} to arrive at
\begin{align}
    &\int_0^t|\mathcal{P}(x(\tau),p(\tau))|d\tau \nonumber\\
    &\leq -\mathcal{L}_{\alpha^*}(x(t),p(t)) + \mathcal{L}_{\alpha^*}(x(0),p(0))\nonumber \\
    &\quad \qquad +\left(\alpha^*+2\bar{\gamma}nc^2\right)\|e^{\lambda A t}\tilde{x}(0)\|_2^2 \nonumber\\
    &\leq \mathcal{L}_{\alpha^*}(x(0),p(0)) +\left(\alpha^*+2\bar{\gamma}nc^2\right)\|e^{\lambda A t}\tilde{x}(0)\|_2^2 \label{ineq:L_bound}
\end{align}
for all $t\geq 0$. Combining \eqref{ineq:L_bound} with the fact that $A$ is Hurwitz, we get
\begin{align}
    \lim_{t\to\infty}\int_0^t|\mathcal{P}(x(\tau),p(\tau))|d\tau < \infty. \label{ineq:lim_P}
\end{align}
Since $\int_0^t|\mathcal{P}(x(\tau),p(\tau))|d\tau$ is increasing in $t$, it follows from \eqref{ineq:lim_P} that $\int_0^t|\mathcal{P}(x(\tau),p(\tau))|d\tau$ has a finite limit as $t\to\infty$.
Additionaly, $\mathcal{F}$ and $\mathcal{T}$ are Lipschitz continuous and $x$ takes values in a compact set. Therefore $x$ and $p$ are uniformly continuous, meaning that $\mathcal{P}(x,p)$ is uniformly continuous. As a result, we can invoke Barbalat's lemma to conclude that $\mathcal{P}(x(t),p(t))\to 0$ as $t\to\infty$. Finally, combining $\lim_{t\to\infty}\mathcal{P}(x(t),p(t))=0$ with $\mathcal{P}(\xi,\mathcal{F}(\bar{\xi}))=0$ if and only if $\xi\in\mathbb{ENE}(\mathcal{F})$, we get $\lim_{t\to\infty}\inf_{\xi\in\mathbb{ENE}(\mathcal{F})}\|x(t)-\xi\|=0$.
\end{proof}

\end{appendix}

\end{document}